\documentclass[fleqn,10pt]{wlscirep}
\usepackage[utf8]{inputenc}
\usepackage[T1]{fontenc}
\usepackage{lineno}
\nolinenumbers

\usepackage{float}
\usepackage{color}
\usepackage{amsmath,amsthm,amssymb}
\usepackage{bbm}
\usepackage{bm}
\usepackage{amsfonts}
\usepackage{mathrsfs}
\usepackage{hyperref}
\usepackage[ruled,boxed,linesnumbered]{algorithm2e}
\usepackage{caption}
\usepackage{subcaption}
\usepackage{graphics}
\usepackage[title]{appendix}
\usepackage{ulem}

\title{City-Scale Holographic Traffic Flow Data based on Vehicular Trajectory Resampling}

\author[1,2]{Yimin Wang}
\author[1,2]{Yixian Chen}
\author[1,2]{Guilong Li}
\author[1]{Yuhuan Lu}
\author[1,3]{Zhi Yu}
\author[1,2,*]{Zhaocheng He}

\affil[1]{Research Center of Intelligent Transportation System, SUN YAT-SEN University, Guangzhou, 510006, PRC}
\affil[2]{Guangdong Provincial Key Laboratory of Intelligent Transportation System, Guangzhou, 510275, PRC}
\affil[3]{Joint Research and Development Laboratory of Smart Policing in Xuancheng Public Security, Xuancheng, 242000, PRC}

\affil[*]{corresponding author: Zhaocheng He (hezhch@mail.sysu.edn.cn)}

\begin{abstract}
Despite abundant accessible traffic data, researches on traffic flow estimation and optimization still face the dilemma of detailedness and integrity in the measurement.
A dataset of city-scale vehicular continuous trajectories featuring the finest resolution and integrity, as known as the holographic traffic data, would be a breakthrough, for it could reproduce every detail of the traffic flow evolution and reveal the personal mobility pattern within the city.
Due to the high coverage of Automatic Vehicle Identification (AVI) devices in Xuancheng city, we constructed one-month continuous trajectories of daily 80,000 vehicles in the city with accurate intersection passing time and no travel path estimation bias.
With such holographic traffic data, it is possible to reproduce every detail of the traffic flow evolution.
We presented a set of traffic flow data based on the holographic trajectories resampling, covering the whole city, including stationary average speed and flow data of 5-minute intervals and dynamic floating car data.
\end{abstract}

\theoremstyle{definition}
\newtheorem{definition}{Definition}[section]
\newtheorem{thm}{Theorem}[section]

\begin{document}

\flushbottom
\maketitle

\thispagestyle{empty}

\noindent Key words: Automatic Vehicle Identification, holographic traffic data, trajectory resampling, virtual traffic measurement.

\section*{Background \& Summary}

The hologram technology \cite{GABOR1948} uses continuous media to record the optical information of objects whose three-dimensional light field can be reproduced afterward. 
Analogously, in this paper, the holographic data of the traffic flow is defined as the global information of all vehicles' dynamics, i.e., the trajectories of each vehicle in the traffic flow.
And the ability to reproduce accurate traffic flow on a city-wide scale has significant implications for real-world traffic control, path planning, and decision-making process. 

Identification devices such as GPS or RFID could record individual trajectories.
However, they can hardly capture all vehicles' trajectories due to their low penetration and spatial sparsity, respectively.
On the other hand, an Automatic Vehicle Identification (AVI) \cite{bernstein1993automatic} device is able to capture the identity and the timestamp of vehicles when passing by a specific checkpoint on the road.
With the growing number of traffic cameras, AVI detectors are implemented in almost every intersection in Chinese cities.
And one can obtain timestamped location sequences of all vehicles benefit from wide distributed AVI detectors on the road network.

With such comprehensive identified traffic data, it is possible to generate the holographic trajectories by enriching details of traffic flow dynamics.
This paper presents a method to reconstruct trajectories of vehicles from discrete serials of AVI observations.
Based on the reconstructed trajectories, we propose a sampling method on traffic flow data to simulate the detecting processes from both views of Eulerian and Lagrangian traffic flow observations, such as traffic count detection by loop detectors and real-time position detection by floating cars.

Moreover, the proposed methods are implemented in Xuancheng, China.
With 97\% of intersections equipped with AVI devices, the system captures almost every vehicular movement on the road network, daily producing 4 million records.
In this case, Xuancheng might be known as the first city empowered with the insight of all-field round-the-clock vehicular trips.
Considering the risk of personal information leaking, researchers are encouraged to collect cross-sectional aggregating data and limited vehicular trajectories through a supervised interactive virtual traffic measurement service.

Such resampled traffic data could support various of transportation-related researches. 
For instance, 1) consistent multi-source detected data could be resampled from the holographic dataset for data fusion research; 2) mobility patterns could be found from full sampled individual trip data; 3) optimal planning of traffic detectors deployment could be tested by placing custom virtual detectors on the data platform.

\section*{Methods}

The AVI technology is widely used in traffic enforcement cameras to automatically identify vehicles involving traffic violations \cite{sun2017simultaneous}, saving numerous human works to recognize license plates from raw images. 
Generally, active AVI detection identifies and records every vehicle passing the checkpoint \cite{yu2013utilizing}, even those not involving traffic violations. 
Thus, each vehicle on the road network would generate a trajectory constituted by a series of identifying records known as license plate recognition (LPR) data \cite{zhan2015lane}. 

However, in the early days, the AVI deployment coverage and license recognition accuracy are not enough to get precise travel paths.
Hence, some of the researches focused on original-destination (OD) reconstruction \cite{Asakura2000, Zhou2006}.
With the significant development of dynamic AVI technology and the wide deployment of AVI cameras, it is possible to reconstruct the intact travel chain using successive LPR records \cite{rao2018origin, khare2019novel}.
Moreover, deep learning algorithms like GNN are employed to reduce uncertainties in identifying vehicles in recent research \cite{Tong2021, li2018diffusion}.

Although the above methods provide plausible solutions to trip reconstruction, path estimation errors are introduced due to the limited AVI coverage.
The estimation accuracy mainly depends on a certain coverage rate, as known as the proportion of AVI-equipped intersections in the whole road network. 
The higher coverage of AVI-equipped intersections implies that there are fewer trip paths to reconstruct.
With the benefit of this high coverage, we could get promising results from some simple and effective reconstruction algorithms.

Therefore, the generic workflow for generating the holographic trajectories and the related resampled data is depicted in Fig \ref{fig:workflow}.
Two main procedures (P1 \& P2) turn the discrete raw LPR data into continuous trajectories through the workflow.
Trip measurement turns the partial observable LPR data into segmental trip data with certain paths on a constructed full-sensing road network (FSRN).
Then trajectory reconstructing interpolation is applied to each segment to form the holographic traffic flow data \cite{wang2018review}.
Finally, one can run virtual traffic detection (P3) on holographic trajectories and resample various traffic flow data.

\begin{figure}[H]
\centering
\includegraphics[width=0.8\textwidth]{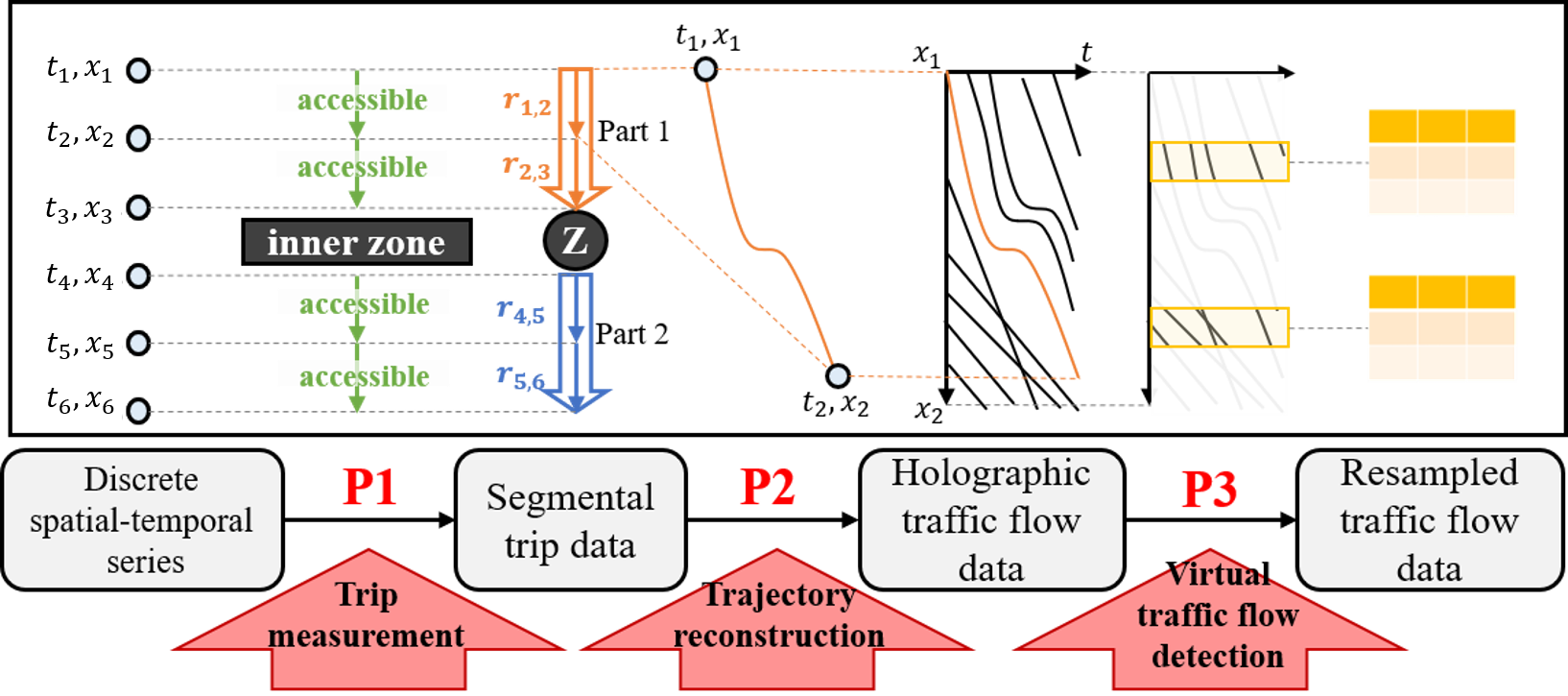}
\caption{Data processing workflow}
\label{fig:workflow}
\end{figure}

\subsection*{Road Network Description}

To avoid path estimation error, the trajectory reconstruction is conducted on a well-defined road network on which the LPR data are mapped.
This paper describes the physical road network (PRN) as a directed graph, denoted as $G^*\left(N^*,S^*\right)$.
The other related notation is in Table. \ref{tab:notations}.
There should be at most one trip path for any serial of LPR records to guarantee no path estimation bias, i.e., $m(A_{s,t}) \in \{0,1\}$.
Let $N^A$ be the set of the AVI-equipped intersections.
It is clear that $N^A \subseteq N^*$.
Assuming an ideal circumstance that $N^A = N^*$, all the trip paths on the physical network can be observed.

When $N^A \subset N^*$, it is still possible to capture all of the trips, as long as the following full-sensing condition is satisfied.

\begin{definition}[Full-sensing road network (FSRN)]
A full-sensing road network (FSRN) is a road network graph that among all the paths between any two different AVI-equipped intersections, there is no more than one path with non-AVI-equipped intersections.
\end{definition}

It guarantees that the path between two consecutive LPR records is determined.
Details of the full-sensing theorem are in Appendix \ref{ap:fsrn}.
This theorem demonstrates that it is unnecessary to deploy an AVI detector on each intersection to get the full-sensing condition.

Let LPR data be $\bm{a}_i=(t_i,x_i)$, containing the timestamp and one-dimensional location of a vehicle passing Node $i$.
Then the record of the trip $r_{1,j}$ consists of a serial of spatial-temporal locations, i.e., $\bm{a}_{1,j}=\{\bm{a}_1,\bm{a}_2,...\bm{a}_i,\bm{a}_j\}$. 
Such as consecutive LPR records $\{\bm{a}_B, \bm{a}_D\}$ in Fig. \ref{fig:deployment}, the path $r_{B,D}=\{B,E,D\}$ can be determined regardless of missing detection.

\begin{figure}[H]
\centering
\includegraphics[width=0.5\textwidth]{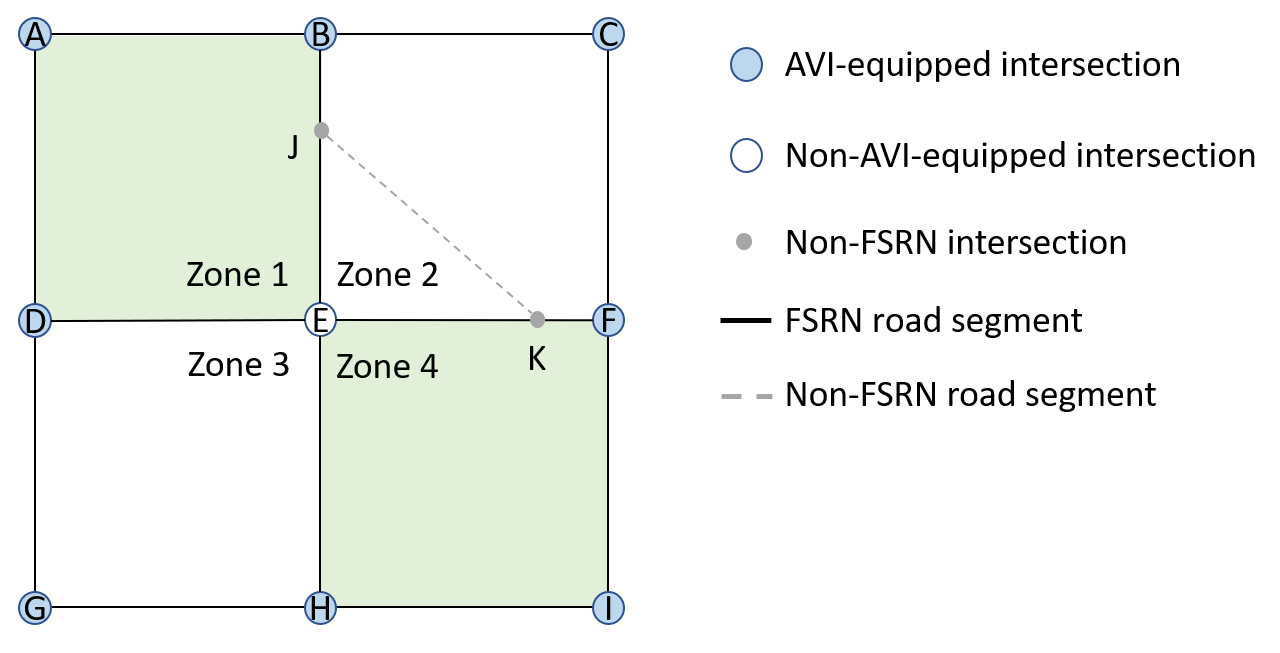}
\caption{Demonstration of the road network and AVI deployment}
\label{fig:deployment}
\end{figure}

Generally speaking, if the PRN fails the full-sensing condition, the challenge is to construct an FSRN according to the locations of AVI-equipped intersections.
The idea is to extract an FSRN from the physical network by eliminating some road segments and intersections.
Then a trip on PRN would be divided into two parts, including on-FSRN parts and off-FSRN parts.
For instance, a trip $r_{A,I}=\{A,B,J,K,F,I\}$ in Fig. \ref{fig:deployment} would be divided into $r_{A,B}=\{A,B\}$, $r_{B,F}=\{B,J,K,F\}$, and $r_{F,I}=\{F,I\}$, where $r_{B,F}$ is the off-FSRN part.
Furthermore, The closed traffic zone is constructed to keep the off-FSRN parts in a particular area.

\begin{definition}[Closed traffic zone]\label{def:zone}
A closed traffic zone is an area bounded by FSRN road segments, and for any non-FSRN segments in the zone, their connected segments are also within the zone area.
\end{definition}

In this way, a trip on the physical road network might be represented by several parts on FSRN separated by staying or mobility within the traffic zones.
The trip $r_{A,I}$ mentioned above could be represented by inter-zone movements $r_{A,B}$, $r_{F,I}$, and inner-zone activity $r_{B,F}$.
Related details can be found in Appendix \ref{ap:zone}.

In order to obtain vehicular movements as high resolution as possible from an AVI-fixed-locating road network, the challenge is to minimize the area of the traffic zones by constructing a suitable sensing network under the constraint of the full sensing criterion.
Additionally, more AVI implemented intersections indicate more resemblance to the FSRN and the PRN. 
Thus more detailed activities can be captured, i.e, $N^A \to N^*$, $FSRN \to PRN$.
A current sensing network of Xuancheng city is shown in Fig. \ref{fig:area}.
In Xuancheng city, the AVI installation rate among the intersections is 97\%.

Despite such an almost ideal trip observation in Xuancheng, the trajectory reconstruction is still a problem of interpretation for observed passing time at both the upstream and downstream ends of a road segment.
For trajectories, the turning directions on each intersection could be easily inferred by downstream LPR records, while their exact lanes are hardly recognized. 
Thus, the traffic flow dynamic would be described by the turning stream on each intersection, rather than different lanes on the road segments \cite{robertson1969transyt}.
For vehicular dynamics within the road segment $s_{i,j}$ of the trip, the trajectory $(t,x_{i,j})$ between $\bm{a}_i$ and $\bm{a}_{j}$ can be calculated as follows: 
\begin{linenomath*}
\begin{equation} \label{eq:int_x}
    x_{i,j}(t) = x_i + \int_{t_i}^{t} v(t)dt, t \in (t_i,t_j), x \in (x_i, x_j)
\end{equation}
\end{linenomath*}
For observation $\bm{a}_{1,j}=\{\bm{a}_1,\bm{a}_2,...\bm{a}_i,\bm{a}_j\}$, let $X=\{x_{1,2},...,x_{i,j}\}$, which represents the set of vehicular trajectories on each segment of trip $r_{1,j}$.
Then the goal is to reconstruct $X$ based on the spatial-temporal trip records $\bm{a}_{1,j}$.

\begin{figure}[H]
\centering
\includegraphics[width = 1 \textwidth]{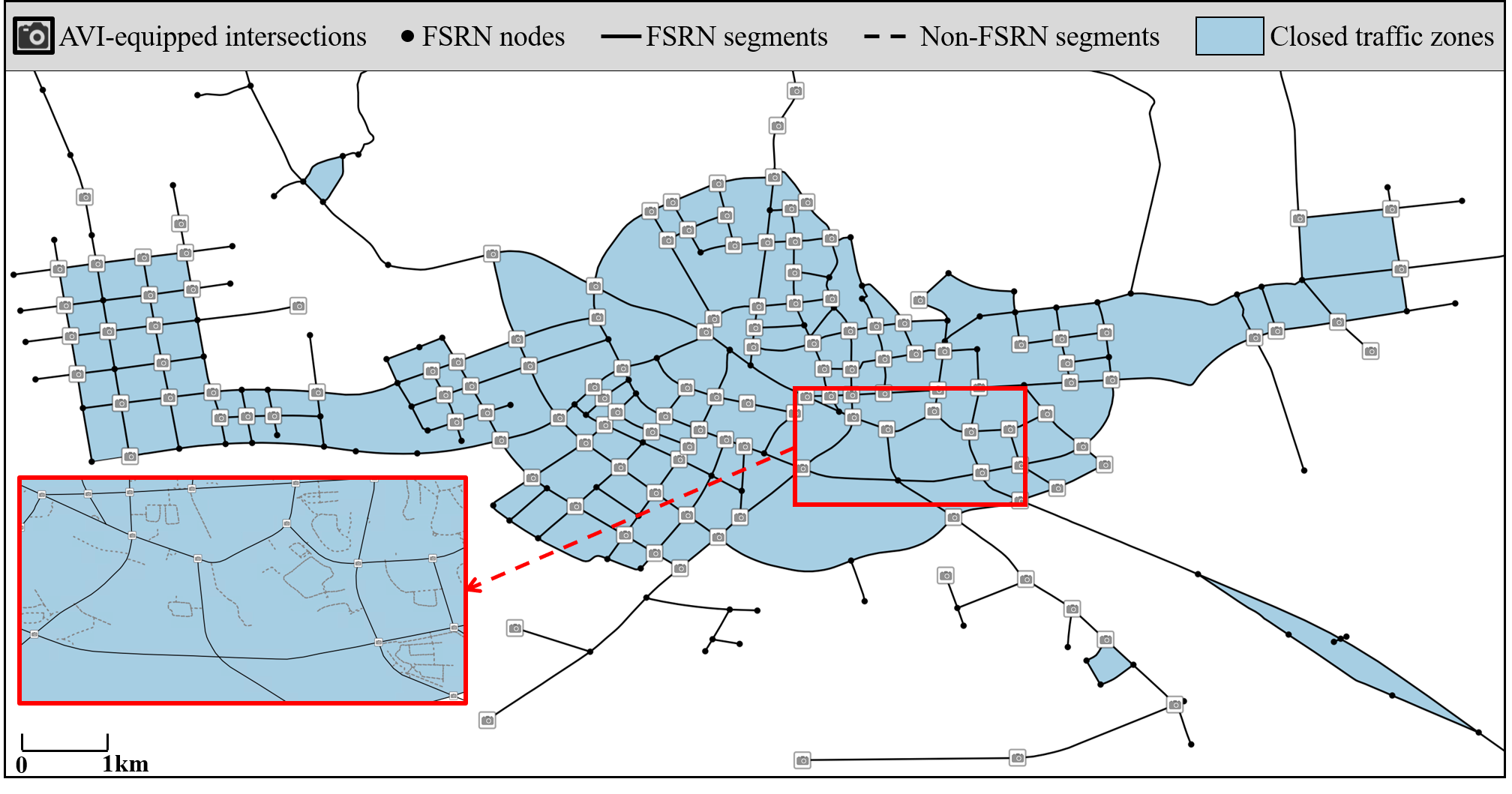}
\caption{Road network and AVI distribution of Xuancheng city}
\label{fig:area}
\end{figure}

\begin{table}[H]
\flushleft
\begin{tabular}{|l|l|}
\hline
Notation & Description \\
\hline
$G^*\left(N^*,S^*\right)$ & The graph of PRN\\ 
\hline 
$s_{i,j}$ & The road segment from Node $i$ to Node $j$\\ 
\hline 
$r_{s,t}=\{s, 1, 2, ..., t\}$ & A trip path from Node $s$ to Node $t$\\ 
\hline
$R_{s,t}$ & The set of $r_{s,t}$\\ 
\hline 
$n(R_{s,t})$ & the number of the possible  trip path from Node $s$ to Node $t$ \\
\hline
$\bm{a}_{s,t}=\{\bm{a}_s,\bm{a}_1,\bm{a}_2,...,\bm{a}_t\}$ & A serial of consecutive LPR data\\ 
\hline 
$m(\bm{a}_{s,t})$ & The number of the possible trip paths for $A_{s,t}$\\
\hline
\end{tabular}
\caption{\label{tab:notations} Description of notations.}
\end{table}

\subsection*{Trip Measurement}

As shown in the workflow (Fig. \ref{fig:workflow}), to get trip-based spatial-temporal serials, a trip dividing algorithm is required.
The basic procedure is to determine whether two consecutive records belong to the same trip. 
In this paper, we use the travel time of a vehicle passing two AVI-equipped intersections $i$ and $j$ as a spatial-temporal accessibility criterion.
\begin{linenomath*}
\begin{equation} \label{eq:Ha}
    H(s_{i,j},t_i,t_j) = \left\{
    \begin{array}{l@{\qquad}l}
        1 & l(s_{i,j}) / v_{min} > (t_j-t_i) \\
        0 & else
    \end{array}
    i,j \in V^A
    \right.
\end{equation}
\end{linenomath*}
where $l(s_{i,j})$ is the length of segment $s_{i,j}$, and $v_{min}$ is the minimal travel speed.
$H=1$ indicates that records $\bm{a}_i$ and $\bm{a}_j$ belong to one trip, while $H=0$ means at least one staying behavior between $\bm{a}_i$ and $\bm{a}_j$.

However, as shown in Fig. \ref{fig:deployment}, not each passing point in trip $r$ can be recorded by AVI detectors, such as $r_{B,D}=\{B,E,D\}$.
In other word, the observation could be a subset of the trip records, i.e.,  $\bm{a}^o = \{\bm{a}_i|i \in N^A\}, \bm{a}^o \subseteq \bm{a}$. 
In between of two observed passing points, the passing time of the non-AVI points shall be inferred in the algorithm.
Taking into account the non-AVI passing points and accessibility criteria in Eq. \ref{eq:Ha}, an algorithm for for trip measurement is proposed ( details in Appendix \ref{ap:pg_alg}).
The idea is that, one can use Eq. \ref{eq:Ha} to judge accessibility on segment $s_{k,k+1}$ between the green light phase $\tau_k=[g_k^{start}, g_k^{end}]$ and $\tau_{k+1}=[g_{k+1}^{start}, g_{k+1}^{end}]$.
\begin{linenomath*}
\begin{equation}
    G(\tau_k, \tau_{k+1})=H(s_{k,k+1},g_k^{end},g_{k+1}^{start})
\end{equation}
\end{linenomath*}
Then we can search accessible downstream green light phases into a set $T_k$ as depicted in Fig. \ref{fig:dn_pgraph}, iteratively.
The downstream searching process runs for $i+1 \leq k \leq j$, and generates the potential passing graph $P_{i,j}(T,E)$ in which the edges indicates two consequent passing phases.
For each accessible phase in layer $T_j$, we can pick the one in which $t_j \in [g_j^{start}, g_j^{end}]$ fits as a proved set of phases $T^*_j$.
As for edges, update the proved edge set as follows.
\begin{linenomath*}
\begin{equation*}
    E^*_{j-1,j} = \{e_{j-i,j} | \tau_j \in T^*_j \}
\end{equation*}
\end{linenomath*}
By updating proved phases and edges in turns from the downstream end to the upstream end, we can trim the graph into a accessible passing graph $P^*_{i,j}(T,E)$ for path from node $i$ to $j$.(Fig. \ref{fig:up_pgrap}).

\begin{figure}[H]
    \centering
    \begin{subfigure}{0.4\textwidth}
         \centering
         \includegraphics[width=\textwidth]{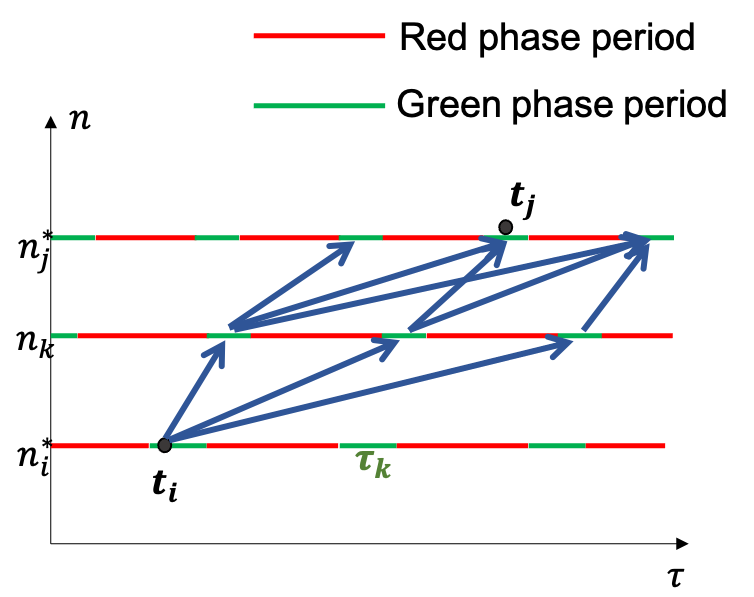}
         \caption{Illustration of searching for potential passing green light phases}
         \label{fig:dn_pgraph}
     \end{subfigure}
     \begin{subfigure}{0.4\textwidth}
         \centering
         \includegraphics[width=\textwidth]{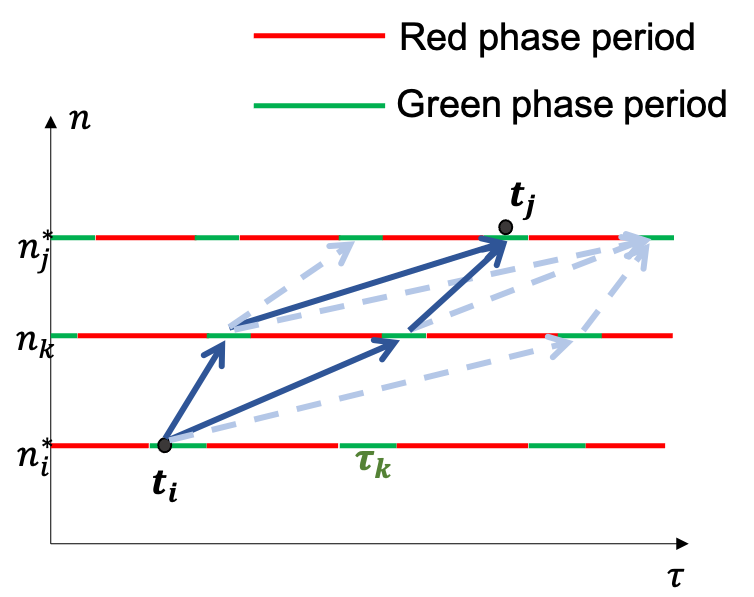}
         \caption{Illustration of accessibility judgment for potential passing green light phases}
         \label{fig:up_pgrap}
     \end{subfigure}
    \caption{Illustrations of inference for passing green light phases of intersections from $i$ to $j$}
    \label{fig:trip_msm}
\end{figure}

Note that AVI detectors might failed recognizing a small portion of the passing vehicles due to poor visual conditions. 
For instance, assuming missing observation $\bm{a}_A$ on trip $\bm{a}=[\bm{a}_B, \bm{a}_A, \bm{a}_D]$ in Fig. \ref{fig:deployment}, the passing-time inference algorithm would be applied for path $[n_B,n_E,n_D]$ since it is the only path between B and D without any AVI-equipped intersections.
If the signals on E did not fit in, such situation would causes trip chain disconnection $(P^*_{i,j}=(\emptyset,\emptyset))$.
Otherwise, it would be a false match.
Therefore, the accuracy of the AVI detection is important to the trip measurement.

\subsection*{Vehicular Trajectories Reconstruction}
The traffic streams consist of the vehicles of the same turning on the road segment.
The dynamics in the same stream would be described as stop-and-go waves caused by the signal periods on the downstream end.

A demonstration of vehicular trajectories in the traffic stream is shown in Fig. \ref{fig:tr_demo}.
The green and red bars on $x=x_j$ represent green and red phases in the signal circles.
Furthermore, the wave's speed is determined by the vehicle queuing state and releasing state of the traffic flow, i.e.,
\begin{linenomath*}
\begin{equation}\label{eq:wave_spd}
    w = - q_m / (k_j-k_m)
\end{equation}
\end{linenomath*}
where $q_m$ is the capacity, $k_m$ is the density under capacity, and $k_j$ is the jammed density. 
In order to calculate vehicular trajectories in Eq. \ref{eq:int_x}, such as the 5 vehicles in Fig. \ref{fig:tr_demo}, the solution of $v(t)$ is formulated as a piecewise function.
\begin{linenomath*}
\begin{equation} \label{eq:piece_spd}
    v(t) = \left\{
                \begin{array}{l@{\qquad}l}
                    v_1 & t_i \leq t < t_1 \\
                    0 & t_1 \leq t < t_2 \\
                    \vdots & \\
                    0 & t_{k-1} \leq t < t_k \\
                    v_j & t_k \leq t \leq t_j
                \end{array}
           \right.
\end{equation}
\end{linenomath*}

\begin{figure}[H]
    \centering
    \includegraphics[width=0.8\textwidth]{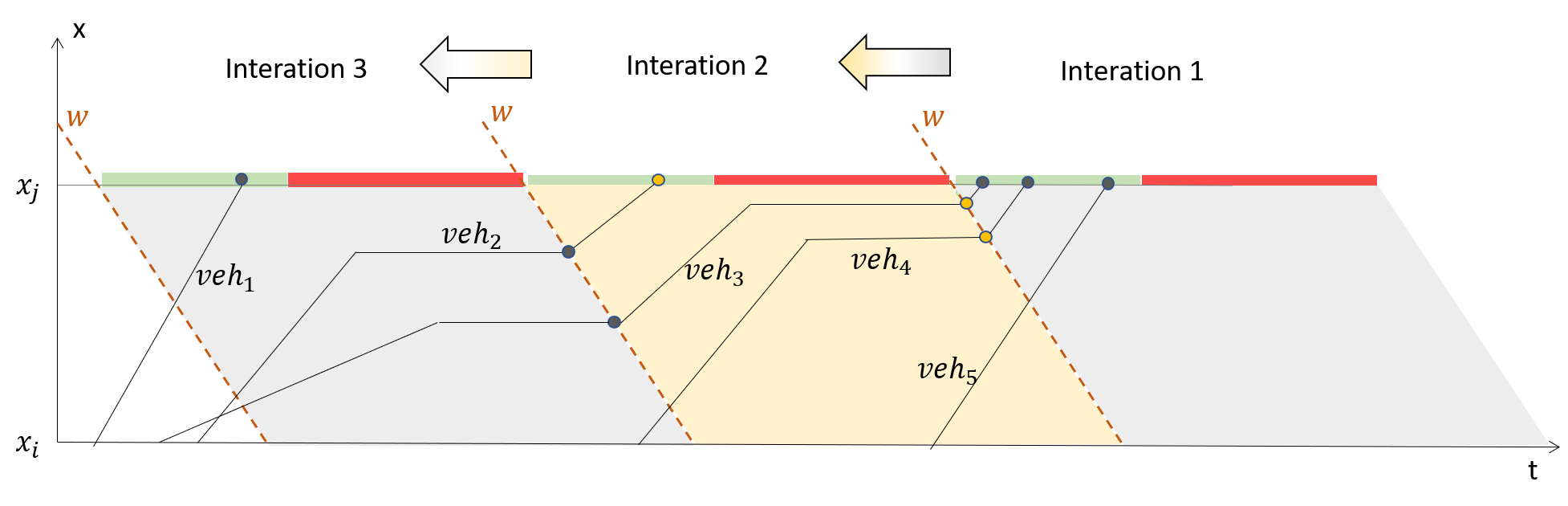}
    \caption{Demonstration of backward trajectory reconstruction from leaving time at $x_j$ to entry time at $x_i$ on a road segment}
    \label{fig:tr_demo}
\end{figure}

To gain the solutions, a backward procedure of trajectories reconstruction is proposed for each passing vehicle, calculating from the downstream to the upstream of the traffic flow.
Hence, the reconstruction begins at the last signal period and iterates by signal circles. 
In other words, the $v(t)$ is calculated from $v_j$ to $v_1$.
Each iteration starts with observations of the passing vehicles in the current period and the remaining ones from the former iteration,  resulting in the new reconstructing states of these vehicles. 
For instance, Iteration 2 in Fig. \ref{fig:tr_demo} contains remained vehicles ($veh_{3,4}$) and passing vehicle(s) ($veh_2$).
At the end of the iteration, $veh_4$'s trajector has been constructed, while trajectories of ($veh_{2,3}$) remained undone and passed to Iteration 3.
The key is to distinguish queued vehicles from non-queued ones.
Then we can complete the trajectories of the non-queued vehicles, leaving the queued ones to the subsequent iterations.
Details of the reconstruction method are in Appendix \ref{ap:tra_rec}.

\subsection*{Virtual Traffic Flow Detection}

With the holistic reconstructed trajectories, the holograph of the city-scale mobility can be acquired.
Note that such a high-resolution individual mobility dataset implies a high risk of personal information being abused.
Thus it is restricted to access the generated raw trajectories directly.
As an alternative, numerical traffic flow detection is applied.
In reality, the traffic flow can be observed from both Eulerian and Lagrangian perspectives. 
Analogously, the reconstructed dataset supports both cross-sectional and vehicular detection.

\begin{figure}[H]
    \centering
    \includegraphics[width=0.6\textwidth]{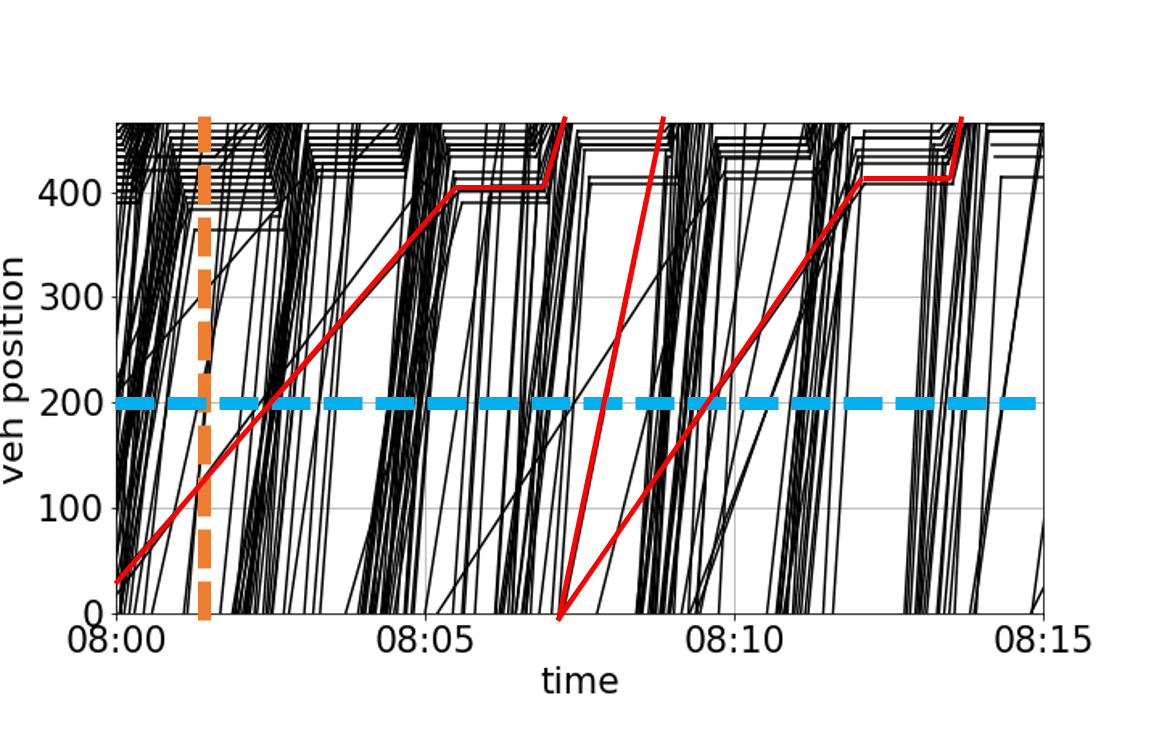}
    \caption{Illustration of virtual traffic flow detection including loop detection (dash line) and floating car detection (red solid line). }
    \label{fig:resamples}
\end{figure}

\subsubsection*{Numerical stationary detection}

For stationary observation, traditional loop data can be simulated by counting intersections of the curves of trajectories crossing the horizontal loop location line as the blue dash line in Fig. \ref{fig:resamples}.
Moreover, the occupancy and velocity can be measured according to the loop's length.
Additionally, segmental measurement could be employed, which detects the instant density (as on purple line) and the space-mean speed (as in orange frame) of the traffic flow as the orange dash line in Fig. \ref{fig:resamples}.
The missing rate is introduced in the loop data resampling process to simulate the systematic detecting error in realistic circumstances.
Each vehicle counting is taken as a Bernoulli trial having the missing rate as the possibility of failure.

\subsubsection*{Virtual floating car detection}

The sample rate controls the penetration of vehicular trajectories resampling, resulting in the red trajectories in Fig. \ref{fig:resamples}.
In order to balance the data utility and personal privacy protection, only the trajectories of commercial vehicles are included in the dataset.
The proportion of commercial vehicles is about 4.5\% to 7\% depends on the time.
Moreover, all of the license numbers are substituted with their unique and irreversible hash code. 

\section*{Data Records}

We provide three types of data to support different research interests:
\begin{itemize}
    \item Short-term anonymized original LPR data
    \item Long-term encrypted reconstructed holographic trajectory data
    \item Long-term resampled traffic data, including loop data and floating car data (FCD) based on the holographic trajectories
\end{itemize}
All of the data are available at the Figshare\cite{figshare2022} repository.

We limit the original LPR data because of the risk of personal information leaking, even if the data are anonymized.
With travel characteristics revealed in the long-termed holographic trajectories, one can still recognize the personal identification using additional data, such as parking lot data.
Hence, it is necessary to encrypt the trajectory data.

However, the long-term resampled traffic data could be used as the primary support for the related research, which could meet most of the needs.
For supplemental use, others can customize their detectors' settings and implement virtual traffic flow detection using the attached resampling software and the encrypted holographic trajectories.
To those interested in the reconstruction method, the short-term anonymized original LPR data could be used for validation.
Details of the three types of data are described as follows.

1) The city-scale loop data and FCD are the one-month long resampling results of the Xuancheng holographic data in Sept. 2020. 
The link-based graph is given in Table. \ref{tab:rdnet} for road network description, including the whole 578 road segments of the city.
The loop dataset provides the 5-minute aggregated flow-speed data, as shown in Table. \ref{tab:loop}.
The floating car dataset includes the trajectories of 500 commercial vehicles are in Table. \ref{tab:fcd}, which is sampled every 10 seconds. 
Their unique IDs can be found in the data repositories.

2) The encrypted holographic trajectories can not be accessed directly; however, one can obtain the self-customized results by using the attached resampling software.
The usage can be found in the following Usage Notes, and the source code of the software is available, see in Code Availability.

3) The short-term original LPR data for reconstruction validation are shown in Table. \ref{tab:lpr}, while the source code of the reconstruction can be found in Code Availability.
The LPR data are collected from 7:00 to 8:00 on a workday morning in Xuancheng.

\begin{table}[H]
\flushleft
\begin{tabular}{|p{2.5cm}|p{14cm}|}
\hline
Column name & Description \\
\hline
ROADID & The ID of road segment, composed of the upstream node ID and the downstream node ID\\ 
\hline 
LANENUM & The number of the lanes of the road segment end\\ 
\hline 
TURN & directions of every downstream road, separated by \# \\ 
\hline 
DN\_ROAD & Road IDs of every downstream road, separated by \#\\ 
\hline 
GEOM & String of geometry objects. \\
\hline 
LEN & Length of road segment in meters. \\
\hline
\end{tabular}
\caption{\label{tab:rdnet} Road network data attributes.}
\end{table}

\begin{table}[H]
\flushleft
\begin{tabular}{|p{2.5cm}|p{14cm}|}
\hline
Column name & Description \\
\hline
ROAD\_ID & The ID of road segment, composed of the upstream node ID and the downstream node ID\\ 
\hline 
FTIME & The beginning timestamp of the interval\\ 
\hline 
TTIME & The ending timestamp of the interval\\ 
\hline 
INT & Data aggregating interval (s)\\ 
\hline 
COUNT & The number of all passing vehicles\\ 
\hline 
REG\_COUNT & The number of regular vehicles\\ 
\hline 
LAR\_COUNT & The number of large vehicles\\ 
\hline 
ARTH\_SPD & The arithmetic mean of vehicle speed (km/h)\\ 
\hline 
HARM\_SPD & The harmonic mean of vehicle speed (km/h)\\ 
\hline 
TURN & The turning direction of the stream, S/L/R/U/Unknown represent straight, left, right, U-turn, and no downstream movements, respectively\\
\hline
\end{tabular}
\caption{\label{tab:loop} Loop data attributes.}
\end{table}

\begin{table}[H]
\flushleft
\begin{tabular}{|p{2.5cm}|p{14cm}|}
\hline
Column name & Description \\
\hline
VID & The ID of vehicles\\ 
\hline 
TYPE & Vehicle types: 1 for large vehicles, 2 for regular vehicles\\ 
\hline 
TIME & Trajectory recorded time \\ 
\hline 
LON & Longitude of the vehicle position\\ 
\hline 
LAT & Latitude of the vehicle position\\ 
\hline 
SPD & Vehicle speed\\ 
\hline 
TURN & The turning direction of the vehicle, S/L/R/U/Unknown represent straight, left, right, U-turn, and no downstream movements, respectively\\ 
\hline 
DIS & Distance from vehicle position to downstream end of the road segment\\ 
\hline 
ROADID & Road segment ID\\ 
\hline 
\end{tabular}
\caption{\label{tab:fcd} Floating car data attributes.}
\end{table}

\begin{table}[H]
\flushleft
\begin{tabular}{|p{2.5cm}|p{14cm}|}
\hline
Column name & Description \\
\hline
VID & The ID of vehicles\\ 
\hline 
FROAD & Road ID of the former passing moment\\ 
\hline 
TROAD & Road ID of the latter passing moment\\ 
\hline 
FTIME & Timestamp of the former passing moment\\ 
\hline 
TTIME & Timestamp of the former passing moment \\ 
\hline
\end{tabular}
\caption{\label{tab:lpr} LPR data attributes.}
\end{table}

\section*{Technical Validation}

The generated traffic flow profile of morning peak is revealed in Fig. \ref{fig:link_flow)}.
The number of passing vehicles is visualized on the map by the width of the blue shades.
It presents the radial distribution of the traffic flow.
To demonstrate the validity of the generated data, we compared the data with different sources to test the consistency in between. 
Also, the characteristics of the generated data are analyzed.
Several data profiles are drawn from the flow-based perspective and trip-based perspective, respectively.

\begin{figure}[H]
    \centering
    \includegraphics[height=8cm]{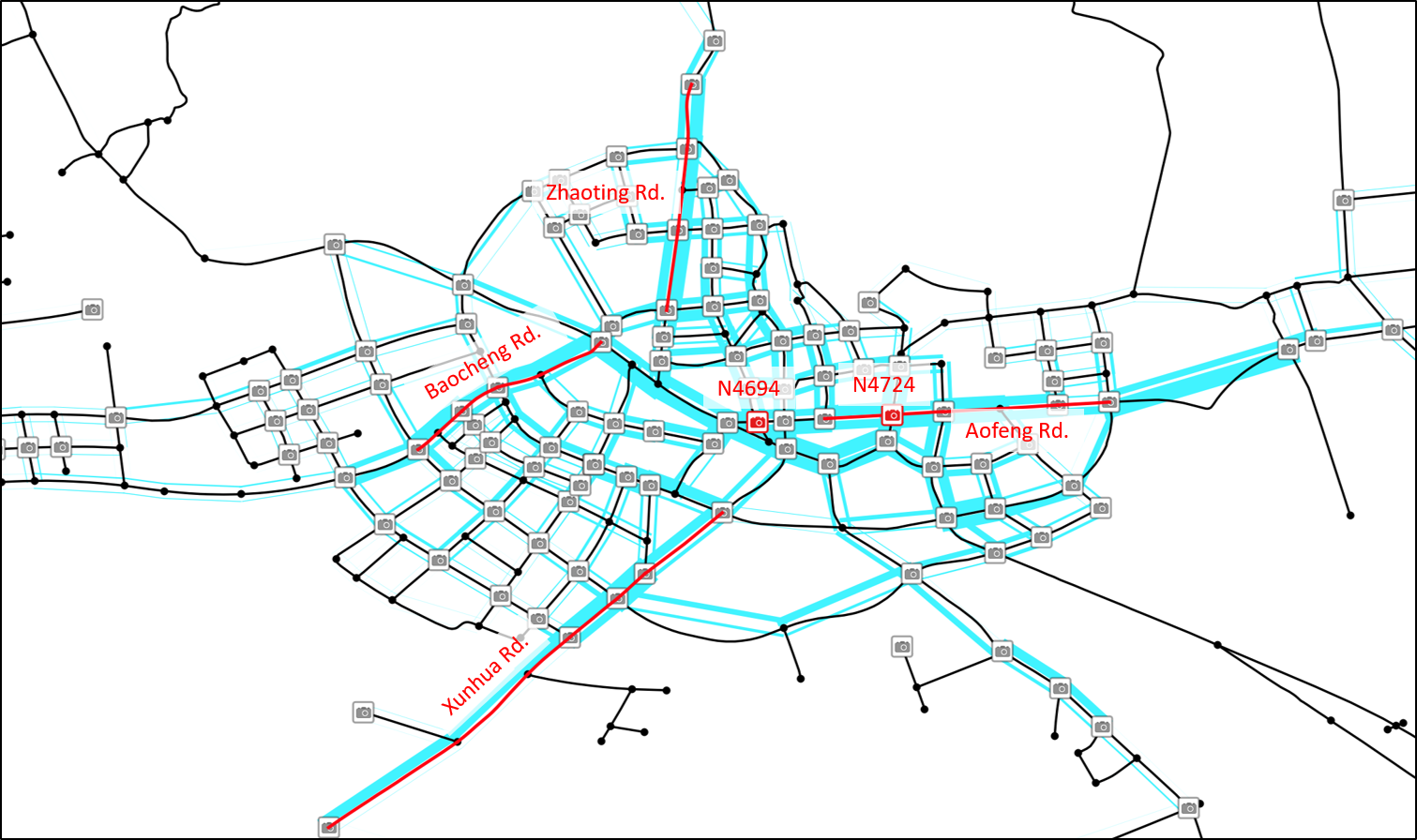}
    \caption{Morning peak traffic in Xuancheng city. The width of the blue shades represents the number of vehicles. Road segments (Zhaoting, Baocheng, Xuanhua, and Aofeng Rd.) and intersections (N4694 \& N4724) in red are to be validated.}
    \label{fig:link_flow)}
\end{figure}

\subsection*{Flow-based Perspective}

The flow-based validation includes comparing the traffic flow data on red-marked roads against another observation and analyzing the generated fundamental diagram.

Fig. \ref{fig:count_valid} depicts the resampled count numbers and the manual results of the southern in-coming stream on intersection N4724.
The resampled data on intersection N4694 and N4724 are compared to the on-sited manual observation, considering vehicles from each in-coming road segment from 11:00 to 12:00 on Sept. $15^{th}$, 2020.
The correlative coefficient is 0.748 with $RMSE = 4.3 veh/min$, which shows the consistency. 

Furthermore, the travel time data on Aofeng Rd., Zhaoting Rd., Baocheng Rd., and Xunhua Rd. are compared to the dynamic estimated results from the Amap.
Since some smooth filters and delays on intersections are usually applied in travel time estimation algorithms, the estimated results are likely different from the raw detected ones.
In this paper, the weekly averaged and zero-mean normalized travel time series are proposed.
Fig. \ref{fig:trt_valid} shows the result on Zhaoting Rd., demonstrating the daily deviation of travel time from the average.
Generally speaking, the overall averaging daily travel time is similar to the estimated result by Amap with the correlative coefficient of 0.749.

As for the fundamental diagram, the profile on Aofeng Rd. is shown in Fig. \ref{fig:fd_valid}.
The resampled data fits the speed-density model in Eq. \ref{eq:sd}, having the capacity of $0.36 veh/s$ and the jam density of $0.19 veh/m$, which indicates that the generated traffic flow is of the same character as the realistic one.

\begin{figure}[H]
    \centering
    \begin{subfigure}{0.33\textwidth}
        \centering
        \includegraphics[width=\textwidth]{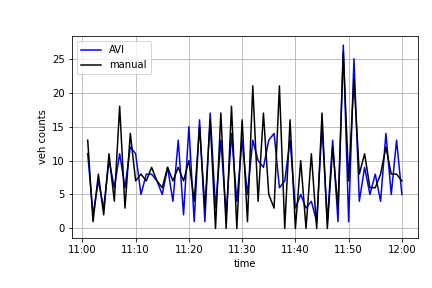}
        \caption{Comparison of traffic counts of southern N4724 (the holographic dataset V.S. on-sited manual observations)}
        \label{fig:count_valid}
    \end{subfigure}
    \begin{subfigure}{0.33\textwidth}
        \centering
        \includegraphics[width=\textwidth]{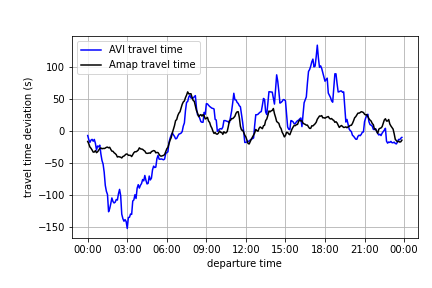}
        \caption{Comparison of travel time on Zhaoting Rd. (the holographic dataset V.S. Amap data)}
        \label{fig:trt_valid}
    \end{subfigure}
    \begin{subfigure}{0.33\textwidth}
        \centering
        \includegraphics[width=\textwidth]{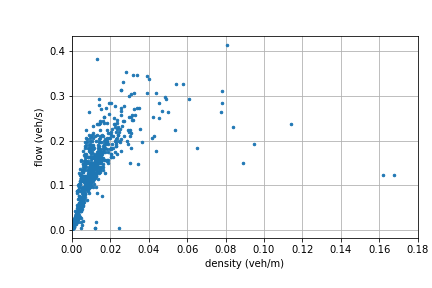}
        \caption{Fundamental diagram of Aofeng Rd.}
        \label{fig:fd_valid}
    \end{subfigure}
    \caption{Validation from the flow-based perspective}
    \label{fig:flow_valid}
\end{figure}

\subsection*{Trip-based Perspective}
The trip-based analysis focuses on the spatial-temporal distribution of the travel demands.
The trip-based analysis is mainly according to the spatial-temporal concentration of the individual trips.
In this paper, the level of spatial concentration of individual travelers is evaluated by the number of different origin-destination zones (ODZ) in a month.
Meanwhile, the level of time concentration is determined by the number of different departure time sections (DTS).
As the individual trip is related to the specific traffic zone surrounded by the road segments, the number of different ODZ is easily counted.
Since departure time is a continuous variable, we conduct a DBSCAN clustering algorithm on each trip to spontaneously generate discrete departure time sections.
To avoid the long tail phenomena of spatial-temporal distribution, we take the $85^{th}$ percentile of the number of DTS and ODZ as the indicators of spatial-temporal concentrating characteristics.

Fig. \ref{fig:dt_pdf} shows the departure time distribution on weekdays of people in different DTS.
One can recognition a typical "Work-Home" commute pattern of those $DTS=2$, which has much higher peaks during commute time.
Besides, the curve of $DTS=4$ seems a "Work-Other-Work-Home" pattern and leads to a midday peak of traffic that does not exist in $DTS=2$ or $DTS=3$ curves.
As for $DTS=3$, there is a noticeable peak at around 20:00 and indicates a "Work-Other-Home" pattern. 
For $DTS=other$, one can find that there are four equivalent peaks at around 7:30, 11:30, 14:00, and 17:30, representing generally high frequent departure times.
Since $DTS=2,3,4$ show the comprehensive mobility patterns, the temporally concentrated travelers are defined as the ones with the $85th$ DTS in $[2,3,4]$.
Note that these patterns have up to four different OD zones.
Likewise, the spatially concentrated travelers are defined as the ones with the $85th$ ODZ less than 5.

Fig. \ref{fig:traveler_cdf} shows the Lorenz curve of travel distance in a month for all travelers, where the cumulative proportion of the travel distance is plotted against the cumulative proportion of individuals \cite{wittebolle2009initial}.
It reveals that mobility distribution on the road network is of the same pattern as other business behaviors.
Among all travelers, the commercial vehicles at the top 1\% of the population share nearly 20\% of the cumulative travel distances.

Some of the trips are predictable due to the traveler's comprehensive characteristics, such as the commuters, the spatially concentrated ones, and the temporally concentrated ones. 
Furthermore, we can estimate the movements of commercial vehicles since they are under surveillance.
These four types of travelers are defined as regular travelers whose patterns are recognizable.

In summary, the regular ones share 37\% of the whole travelers but form 45\% of the whole travel distance (Fig. \ref{fig:pie}).
Thus, once these 37\% regular travelers are well modeled, we can reproduce nearly half of the trips, and the other half might be generated with random methods.

\begin{figure}[H]
    \centering
    \begin{subfigure}{0.36\textwidth}
        \centering
        \includegraphics[width=\textwidth]{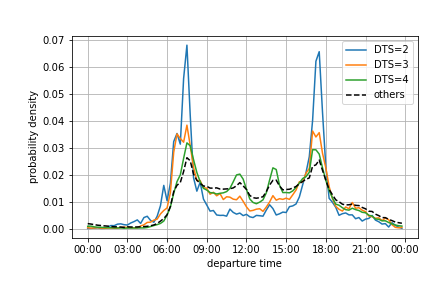}
        \caption{Departure time distribution for different mobility patterns}
        \label{fig:dt_pdf}
    \end{subfigure}
    \begin{subfigure}{0.24\textwidth}
        \centering
        \includegraphics[width=\textwidth]{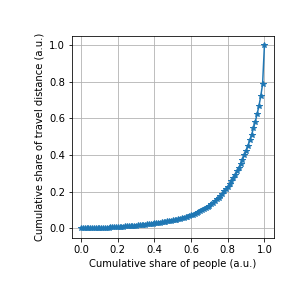}
        \caption{Lorenz curve of travel distance within a month for all travelers}
        \label{fig:traveler_cdf}
    \end{subfigure}
    \begin{subfigure}{0.36\textwidth}
        \centering
        \includegraphics[width=\textwidth]{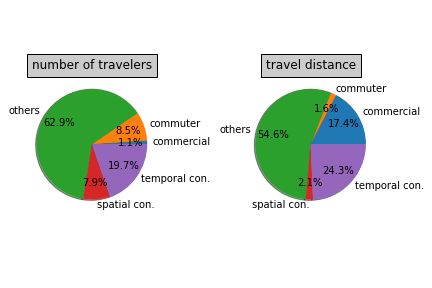}
        \caption{Pie chart of population and travel distance for different travelers, including commuters, commercial vehicles, temporally concentrated travelers, and spatially concentrated travelers}
        \label{fig:pie}
    \end{subfigure}
    \caption{Validation from the trp-based perspective}
    \label{fig:trip_prof_all}
\end{figure}

\section*{Usage Notes}

As mentioned above, there are three types of data we provide. 
The short-term LPR data and long-term resampled traffic data can be downloaded for static data usage.
On the other hand, the encrypted holographic trajectories can be used in the interactive measurement of the traffic flow.
Users can modify the virtual detecting environment and get customized virtual detection results.
In this way, we can offer the user-customized round-the-clock long-term traffic flow data to the most satisfactory resolution without exposing personal trajectories.

\subsection*{Static Dataset Usage}
 
The road network file can be imported into the PostGIS database or other supported GIS systems through QGIS.
The loop data of each road segment can be used for studying large-scale traffic data prediction.
By combining floating car data with the loop data, users could examine various data fusion models.
Moreover, the FCD data process script could help aggregate individual floating car samples into the segmental travel time.
As for LPR data, each row of the dataset is a pair of consecutive records captured by the AVI detectors.
One can rebuild the route between these two records with the road network.

\subsection*{Interactive Measurement Usage}

The resampling software is a command-line tool to implement virtual traffic flow detection in encrypted trajectories.
Users could tweak the settings in the running properties file and get resampled traffic data straight in the local output files.

In the properties file, users can set the road sections ("ftNode") and time ("fTime", "tTime") of the measurement and define the parameters of loop and floating car detection.
Users can switch on or off the floating car detection by setting the "needFCD" property to "true" or "false".
Furthermore, "fcdSamplingSec" denotes the FCD's sampling period (seconds).
For loop detectors, they are identified by the ID ("loopId"), detecting on the specified road segment ("ftNode").
The loop's position is determined by the property "position", which denotes the distance from the downstream end of the road.
The missing rate ("missingRate") and the aggregating interval ("interval") settings are available.

The software can run on Linux, Windows, and macOS systems using different launchers.
The command is simple as "osLauncher java -jar /path/to/resampling\_software -d /path/of/holoData -c /path/of/properties\_file".

Other details can be found in the "README" file.

\section*{Code availability}
To further describe the details of data processing in our method, we also provide code and instructions for reproducing the presented results \cite{githubrepo}.
In general, files that end with ".py" are supporting python module files, other files with ".ipynb" are written as Jupyter Notebook instruction, and the files under the folder "measurement" are the source code of the resampling software. 
The instruction files demonstrate the whole data processing workflow in Fig. \ref{fig:workflow}, including trip measurement, trajectory reconstruction, virtual traffic flow detection, and data validation.
These files can be used to better understand the modeling and validation steps.


\begin{appendices}
\section{Full-sensing theorem}\label{ap:fsrn}
Among all the paths between any two different AVI intersections in the study area, if there is no more than one path with non-AVI-equipped intersections, then the trip path for the LPR record is determined, i.e., 

\begin{thm}
$\forall i,j \in N^A$, Let $R = \{r_{i,j} | r_{i,j} \cap N^A = \{i,j\}\}$. If $n(R_{i,j}) \in \{0,1\}$, then $\forall p,q \in N^*$, $m(A_{p,q}) \in \{0,1\}$
\end{thm}

\begin{proof}
\begin{align*}
    \displaystyle m(A_{p,q}) & = \prod_{i,j \in A_{p_q}} n \left(R_{i,j}\right) \\
    \because      n          & \in \{0,1\} \\
    \therefore    \prod n    & \in \{0,1\}
\end{align*}
\end{proof}

\section{Closed zone theorem}\label{ap:zone}
If the traffic zone area is bounded by FSRN road segments, and for any non-FSRN segments in the zone, their connected segments are also within the zone area, then the trip of the physical road network (PRN) can be represented as parts on full-sensing road network (FSRN) separated by inner zone activities, i.e.,

\begin{thm}
Let $r^*_{o,d}=\{o, i, i+1, i+2, ..., i+m, d\}$ be a trip on a physical road network, and $Z$ a closed traffic zone on the corresponding full-sensing road network that $s_{i,i+1} \subset \Bar{Z}$.
$\forall m \geq 1$, $s_{i+m-1, i+m}$ on non-FSRN,
then $\forall m \geq 1$, $s_{i+m-1, i+m} \subset \Bar{Z}$.
($\Bar{Z}$ denotes the closure of area $Z$.)
\end{thm}

\begin{proof}
Suppose $s_{i+k-1, i+k} \subset \Bar{Z}$. According to Definition \ref{def:zone},
\begin{align*}
    s_{i+m-1, i+m}             &\subset \Bar{Z}, (m = k) \\
    \because s_{i+k-1, i+k}    &\cap s_{i+k, i+k+1} = \{i+k\} \neq \emptyset \\
    \therefore s_{i+k, i+k+1}  &\subset \Bar{Z} \\
    \Rightarrow s_{i+m-1, i+m} &\subset \Bar{Z},(m = k+1) \\
    \because s_{i,i+1}         &\subset \Bar{Z}, (m=1) \\
    \Rightarrow s_{i+m-1, i+m} &\subset \Bar{Z}, (m \geq 1)
\end{align*}
\end{proof}

\section{Passing-time inference algorithm}\label{ap:pg_alg}
\begin{algorithm}[H]
 \SetAlgoLined
  \KwResult{Accessible passing graph $P^*_{i,j}(T,E)$ }

  $T^*=\emptyset$, $E^*=\emptyset$, $P^*_{i,j}=(T^*,E^*)$ \;
  \eIf{$s_{i,j} \in S^*$}{
      \uIf{$H(s_{i,j},t_i,t_j)=1$}{
        $T^* \gets T^* \cup \{\tau_i,\tau_j\}$, $E^* \gets E^* \cup \{e_{i,j}\}$ \;
      }
  }{
    get path $r=[n_i,n_k,...n_j | i,j \in V^A, k,k+1,... \notin V^A]$\;
    $k \gets k$\;
    $T_{k-1}= \{ \tau_i \}$, $t_i \in \tau_i $\;
    \While{$k \leq j$}{%
      $T_k = \{\tau_k | G(\tau_{k-1},\tau_k)=1, \tau_{k-1} \in T_{k-1} \}$\;
      $E_k = \{e_{k-1,k} | G(\tau_{k-1},\tau_k)=1, \tau_k \in T_k, \tau_{k-1} \in T_{k-1} \}$\;
      $k \gets k + 1$\;
    }
    
    $T^*_k \gets \{\tau_j\}$, $t_j \in \tau_k$\;
  
    \While{$ k > i + 1 $}{%
      update proved edges by $E^*_{k-1,k} = \{e_{k-1,k} | e_{k-1,k} \in E_{k-1,k}, \tau_k \in T^*_k \}$ \;
      update proved phases by $T^*_{k-1} = \{ \tau_{k-1} | \tau_{k-1} \in T_{k-1}, (k-1) \in E^*_{k-1,k} \}$ \;
      $T^* \gets T^* \cup T^*_k$, $E^* \gets E^* \cup E^*_k$, $k \gets k - 1$\;
    }
  }
    
  \caption{Passing Time Inference}
\end{algorithm}

\section{Details of trajectory reconstruction}\label{ap:tra_rec}
As shown in Fig. \ref{fig:traj_rc_0}, there are two different circumstances we need to deal with when it comes to queuing discrimination.
The common idea is that the low constructed travel speed assumes a queuing behavior since the vehicle does not move during the queuing process.
For those vehicles leaving $x_j$, the travel speed is simply determined by the slope between the entry point (A) and leaving point (B), as shown in Fig. \ref{fig:traj_rc_0}.
As for vehicles from former iterations, since the exit point (G) remains unknown, the intersection ($F$) of the wave $\mu_{\tau}$ and stopping position $\overline{FH}$ is chosen as the referring point.
Hence the adapted travel speed is related to Point E and Point F.
Especially when it provides the green light period instead of the exact entry time, the end of the green light period is used as referring point. 

After the independent queuing discrimination, the result might show that several vehicles are assumed queuing before the current green light period.
For instance, let Vehicle 1,3 be the low-speed vehicles as depicted in Fig. \ref{fig:traj_rc_3}.
It is a fact that there is no more than one stop wave during one signal period.
Thus, the queuing vehicle must be in front of the other ones.
Considering the one-wave constraint, let the last low-speed vehicle be the last queuing vehicle.
In this case, Vehicle 1,2,3 would be marked as the queuing vehicles.
Their stopped positions are calculated according to their leaving orders.
The stop position of the i-th vehicle is formulated as follows,
\begin{linenomath*}
\begin{equation}\label{eq:stop_order}
    P^i_\tau = \frac{i-1}{k_j}
\end{equation}
\end{linenomath*}
where $k_j$ is the jam density.
The passing speed is related to the stopped position and the exit point.
On the other hand, the travel speed of the non-queued vehicles is calculated according to the passing information.
The reconstructed trajectory is the straight passing line to vehicles with specific entry and leaving points, such as vehicle 4.
For vehicles with one exact passing point, such as vehicles 5 and 7, the travel speed is formulated by the speed-density model \cite{may1967},
\begin{linenomath*}
\begin{equation}\label{eq:sd}
    v = v_f \cdot \left( 1 - \left(\frac{k}{k_j}\right)^\beta \right)^\alpha
\end{equation}
\end{linenomath*}
where $v_f$ is the free flow speed, and $\alpha=1.0$, $\beta=0.05$ according to relating researches \cite{Ben-Akiva2001a,Xu2014b}.
In this way, the travel speed is given based on the local density, representing the road segment's traffic dynamic.
Then their trajectories are fixed by one passing point and the running speed.
Finally, to vehicles without exact observations, their speed is also calculated by the same speed-density model, and the endpoint is given randomly with constraints of the proceeding and following vehicles. (See Vehicle 4 in Fig. \ref{fig:traj_rc_3}.)

\begin{figure}[H]
    \centering
    \begin{subfigure}{.49\textwidth}
        \centering
        \includegraphics[width=\textwidth]{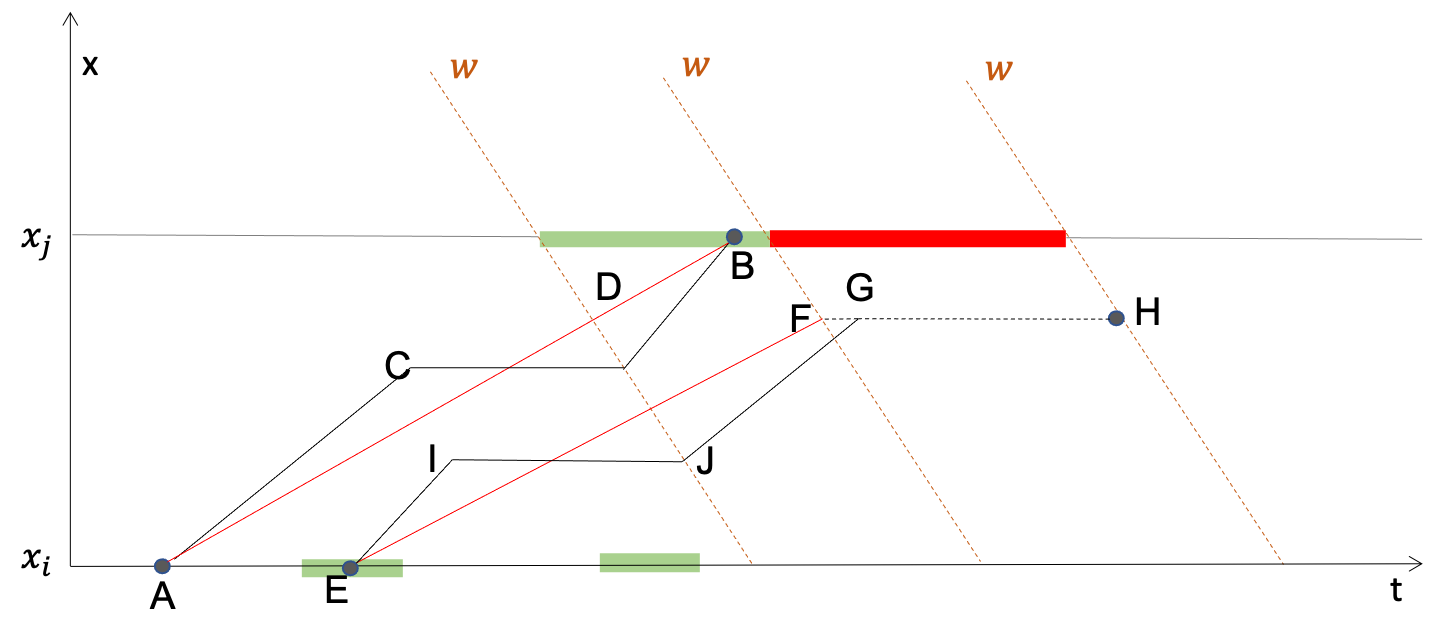}
        \caption{Illustration of queuing discrimination}
        \label{fig:traj_rc_0}
    \end{subfigure}
    \begin{subfigure}{.49\textwidth}
        \centering
        \includegraphics[width=\textwidth]{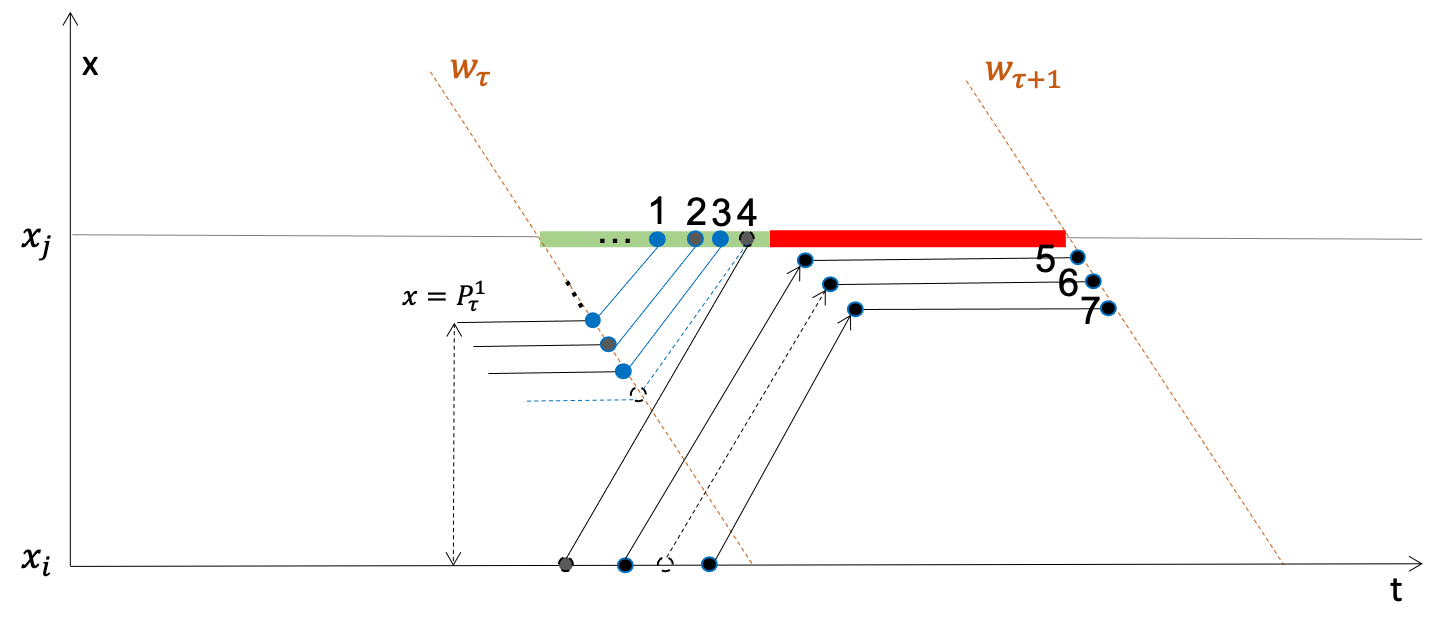}
        \caption{Illustration of queuing tail discrimination}
        \label{fig:traj_rc_3}
    \end{subfigure}
    \caption{Details of trajectory reconstruction}
    \label{fig:traj_rc}
\end{figure}

\end{appendices}

\section*{Acknowledgements}


The work was done at the SYSU Research Center of ITS in the context of the collaboration with the Joint Research and Development Laboratory of Smart Policing in Xuancheng Public Security.
The research was also supported by the National Natural Science Foundation of China (No. U1811463).

\section*{Author contributions statement}

Z.Y. conceived of the presented idea. 
Y.W. developed the theoretical framework and performed the computations.
Y.C and Y.L. contributed to the technical details of the the theory.
G.L. conducted part of the experiments.
Z.H. supervised the findings of this work. 
All authors discussed the results and contributed to the final manuscript.

\section*{Competing interests}

The authors declare no Competing interests.

\end{document}